\documentclass[11pt]{article}
\usepackage[letterpaper]{geometry}
\usepackage{amsmath,amsthm,amsfonts,amssymb}
\usepackage{enumerate,color,xcolor}
\usepackage{graphicx}
\usepackage{subfigure}
\usepackage[square,numbers]{natbib}
\usepackage{ulem}
\usepackage{cancel}

\numberwithin{equation}{section}
\theoremstyle{plain}

\newtheorem{theorem}{Theorem}

\newtheorem{proposition}{Proposition}

\begin{document}

\begin{center}
  \Large \bf A State-Dependent Dual Risk Model
\end{center}

\author{}
\begin{center}
  Lingjiong Zhu\,\footnote{Department of Mathematics, Florida State University, 1017 Academic Way, Tallahassee, FL-32306, United States of America; Email: zhu@math.fsu.edu; Tel.: 1-850-644-3800; Fax: 1-850-644-4053.
  }
\end{center}

\begin{center}
 \today
\end{center}




\begin{abstract}
In a dual risk model, the premiums are considered as the costs and the claims are regarded
as the profits. The surplus can be interpreted as the wealth of a venture capital, 
whose profits depend on research and development. In most of the existing literature of dual risk models,
the profits follow the compound Poisson model and the cost is constant. In this paper, we develop
a state-dependent dual risk model, in which the arrival rate of the profits and the costs depend
on the current state of the wealth process. Ruin probabilities are obtained in closed-forms. 
Further properties and results will also be discussed.
\end{abstract}


\section{Introduction}

The classic risk model is based on the surplus process $U_{t}^{\ast}=u^{\ast}+\rho^{\ast} t-\sum_{i=1}^{N^{\ast}_{t}}C^{\ast}_{i}$,
where the insurer starts with the initial reserve $u^{\ast}$ and receives the premium at a constant rate $\rho^{\ast}$, 
and $\sum_{i=1}^{N^{\ast}_{t}}C^{\ast}_{i}$ is a compound Poisson process where $N^{\ast}_{t}$
is the number of claims till time $t$ and $C_{i}^{\ast}$ is the size of the $i$-th claim. A central problem is to study the ruin probability that the surplus process will ever hit zero, i.e. $\mathbb{P}(\tau^{\ast}<\infty|U_{0}^{\ast}=u^{\ast})$,
where $\tau^{\ast}=\inf\{t>0:U_{t}^{\ast}\leq 0\}$. In recent years, 
a dual risk model has attracted many attentions, see e.g. equation (1.1) in Afonso et al. \cite{Afonso}, 
in which the surplus process is modeled as
\begin{equation}
U_{t}=u-\rho t+\sum_{i=1}^{N_{t}}C_{i},
\end{equation}
where $C_{i}$ are i.i.d.
positive random variables distributed according to $Q(dc)$ independent
of $N_{t}$, which is a Poisson process with intensity $\lambda$.
We assume $\lambda\mathbb{E}[C_{1}]>\rho$.
The surplus can be interpreted as the wealth of a venture capital, 
whose profits are driven by the investment on research and development. 
The profits are uncertain and modeled as a compound jump process (i.e. $\sum_{i=1}^{N_{t}}C_{i}$)
and the costs are more predictable and are modeled as a deterministic process (i.e. $\rho t$).
The company pays expenses continuously over time for the research and development and gets profits
at random discrete times in the future.
Many properties have been studied for the dual risk model.
The ruin probability $\psi(u)=\mathbb{P}(\tau<\infty|U_{0}=u)$, 
where 
\begin{equation}
\tau=\inf\{t>0:U_{t}\leq 0\}, 
\end{equation}
satisfies
the equation, see e.g. Afonso et al. \cite{Afonso}
\begin{equation}
\psi(u)=e^{-\lambda\frac{u}{\rho}}
+\int_{0}^{\frac{u}{\rho}}\lambda e^{-\lambda t}
\int_{0}^{\infty}\psi(u-\rho t+c)Q(dc)dt.
\end{equation}
It is well known that $\psi(u)=e^{-\bar{\alpha} u}$ where $\bar{\alpha}$ is the unique positive solution to the equation:
\begin{equation}
\lambda\left(\int_{0}^{\infty}e^{-\bar{\alpha} x}Q(dx)-1\right)=-\rho\bar{\alpha}.
\end{equation}
When there is a random delay
for the innovations turned to profits, the dual risk model
becomes time inhomogeneous and the ruin probabilities and the ruin time distributions
are studied in \cite{ZhuDelayed}. 

Avanzi et al. \cite{Avanzi} worked on optimal dividends in the dual risk model where
the wealth process follows a L\'{e}vy process and the optimal strategy is a barrier strategy. 
Albrecher et al. \cite{Albrecher} studied a dual risk model with tax payments.
For general interclaim time distributions and exponentially distributed $C_{i}$'s, an expression
for the ruin probability with tax is obtained in terms of the ruin probability without taxation. 
When the interclaim times are exponential or mixture of exponentials, explicit expressions are obtained.
Ng \cite{Ng} considered a dual model with a threshold dividend strategy, with exponential interclaim times.
Afonso et al. \cite{Afonso} worked on dividend problem in the dual risk model, assuming exponential interclaim times.
They presented a new approach for the calculation of expected discounted dividends.
and studied ruin and dividend probabilities, number of dividends, time to a dividend.
and the distribution for the amount of single dividends.
Avanzi et al. \cite{AvanziII} studied a dividend barrier strategy for the dual risk model 
whereby dividend decisions are made only periodically, but still allow ruin to occur at any time.
Cheung \cite{CheungI} studied the Laplace transform of the time of recovery after default, amongst other concepts
for a dual risk model. 
Cheung and Drekic \cite{CheungII} studied dividend moments in the dual risk model. They derived 
integro-differential equations for the moments of the total discounted dividends which can be solved explicitly 
assuming the jump size distribution has a rational Laplace transform.
Rodr\'{i}guez et al. \cite{RCE} 
worked on a dual risk model with Erlang interclaim times, studied the ruin probability, the Laplace transform
of the time of ruin for generally distributed $C_{i}$'s. They also studied the expected discounted dividends assuming
the profits follow a Phase Type distribution. When the profits are Phase Type distributed, Ng \cite{NgII} also studied
the cross probabilities. Yang and Sendova \cite{YS} studied the Laplace transform of the ruin
time, expected discounted dividends for the Sparre-Andersen dual model.
The dual risk model has also been used in the context of venture capital investments and some optimization problems
have been studied, see e.g. Bayraktar and Egami \cite{BE}.
In Fahim and Zhu \cite{FZ}, they studied the optimal control problem for the dual risk model, which is the minimization
of the ruin probability of the underlying company by optimizing over the investment in 
research and development. They also obtained the asymptotic analysis for the dual risk model in \cite{FZII}.

In this paper, we develop a state-dependent dual risk model. 
In the dual risk model, the surplus can be considered
as the capital of an economic activity such as research and development
where gains are random, at random instants, and costs are predictable; see
e.g. the discussions in \cite{Avanzi,Afonso}.
However, we argue that a state-dependent dual risk model,
in which the jump process and the costs are state-dependent may
better reflect the real world applications.
The innovations of a company may have self-exciting
phenomena, i.e., an innovation or breakthrough will increase the chance of the next innovation and breakthrough.
Also, when the wealth process increases, the company will be in a better shape to innovate and hence the arrival rate
of the profits, may depend on the state of the wealth rather than simply being Poisson. Also, the expenses
that a company pays for research and develop may also increase after the company receive more profits. 
For the high tech and fast-growing companies, the running cost and the revenues of a company
grow in line with the size of the company, see e.g. Table~\ref{Alphabet}, where we considered
the annual total revenues, cost of total revenues and the gross profits\footnote{Gross profit is the difference
between the revenue and the cost of the revenue.} in the years 2018-2021
\footnote{Available at \texttt{https://www.macrotrends.net/stocks/charts/GOOG/alphabet/revenue}.}.
We can see the upward trend of growth for Alphabet. Therefore, for a high tech company
such as Alphabet, the usual constant assumption for running cost, the intensity of profits arrivals
in the dual risk model might be too simplistic. 
On the other hand, for a traditional blue-chip company such as Coca-Cola, 
the annual total revenues, cost of total revenues and the gross profits do not vary 
too much year over year, see e.g. Table~\ref{KO}, where we
considered the annual total revenues, cost of total revenues and the gross profits in the years 2018-2021
\footnote{Available at \texttt{https://www.macrotrends.net/stocks/charts/KO/cocacola/revenue}.}.
That might also be the pattern for a high tech company 
that has already matured and no longer has stellar growth. 
Therefore, the dual risk model in the existing literature might be a good model
when the financials of a company do not change too much over time. 
A state-dependent dual risk model might be more appropriate when
the underlying company has phenomenal growth. 

\begin{table}[htb]
\centering 
\begin{tabular}{|c|c|c|c|c|} 
\hline 
Full Year & 2018 & 2019 & 2020 & 2021 
\\
\hline
Revenue (millions) & \$136,819 & \$161,857 & \$182,527 & \$257,637
\\
Cost of Revenue (millions) & \$59,549 & \$71,896 & \$84,732 & \$110,939
\\
Gross Profit (millions) & \$77,270 & \$89,961 & \$97,795 & \$146,698
\\
\hline 
\end{tabular}
\caption{Revenue and Cost by Alphabet during 2018-2021.}
\label{Alphabet} 
\end{table}

\begin{table}[htb]
\centering 
\begin{tabular}{|c|c|c|c|c|} 
\hline 
Full Year & 2018 & 2019 & 2020 & 2021 
\\
\hline
Revenue (millions) & \$34,300 & \$37,266 & \$33,014 & \$38,655
\\
Cost of Revenue (millions) & \$13,067 & \$14,619 & \$13,433 & \$15,357
\\
Gross Profit (millions) & \$21,233 & \$22,647 & \$19,581 & \$23,298\\
\hline 
\end{tabular}
\caption{Revenue and Cost by Coca-Cola during 2018-2021.}
\label{KO} 
\end{table}

Therefore, it will be reasonable to assume that the costs depend on the state of the wealth of the company. 
Indeed, it is not only possible that the company spends more capital on research and development when the profits increase, 
it is also quite common in the technology industry to increase the capital spending on research
when the company is lagging behind its pairs so that it is fighting for survival and catch-up. 
When we assume that the cost is constant, the wealth process of the company is illustrated Figure~\ref{WealthI}(a) till
the ruin time. If we allow the cost to depend linearly on the wealth, the wealth process of the company
is illustrated in Figure~\ref{WealthI}(b). When the dual risk model uses the classical compound Poisson
as the wealth process, the probability that the company eventually ruins decays exponentially in terms of the initial
wealth of the company. As we will see later in the paper, e.g. Figure~\ref{RuinProb}, by
allowing the costs and arrival rates of the profits depending on the state of the wealth process, the model
becomes much more robust, and the ruin probability can decay superexponentially in terms of the initial wealth, i.e., Figure~\ref{RuinProb}(a), Table~\ref{RuinTable},
and it can also decay polynomially in terms of the initial wealth, i.e., Figure~\ref{RuinProb}(b), Table~\ref{RuinPolyTable}.

We are interested to develop a state-dependent dual risk model, in which both costs and the jump process
are state dependent, in contrast to the existing literature on dual risk models, where the jump process is always
assumed to be a Poisson process.
Let us assume that the wealth process $U_{t}$ satisfies the dynamics
\begin{equation}\label{UDynamics}
dU_{t}=-\eta(U_{t})dt+dJ_{t},\qquad U_{0}=u,
\end{equation}
where $J_{t}=\sum_{i=1}^{N_{t}}C_{i}$ and $N_{t}$ is a simple point process with intensity $\lambda(U_{t-})$
at time $t$. Here, $\eta(\cdot):\mathbb{R}_{\geq 0}\rightarrow\mathbb{R}_{\geq 0}$ 
and $\lambda(\cdot):\mathbb{R}_{\geq 0}\rightarrow\mathbb{R}_{\geq 0}$
are both continuously differentiable. Throughout the paper, unless specified otherwise, we assume
that $C_{i}$ are i.i.d. exponentially distributed with parameter $\gamma>0$. 
While allowing $\eta(\cdot)$ and $\lambda(\cdot)$ to be general, the drawback 
of our model is that we restrict $C_{i}$'s to be exponentially distributed for the paper.
It will be an interesting future research project to investigate generally distributed $C_{i}$'s. 
For the wealth process $U_{t}$ in \eqref{UDynamics}, we will obtain closed-form
expressions for the ruin probability and further properties will also be studied.

It is worth noting that the $U_{t}$ process in \eqref{UDynamics} is an extension
of the Hawkes process with exponential kernel and exponentially distributed jump sizes, (see e.g. \cite{Hawkes,Bremaud}) 
that is, a simple point process $N_{t}$ with intensity 
\begin{equation*}
\lambda\left(ue^{-\beta t}+\sum\nolimits_{i:\tau_{i}<t}C_{i}e^{-\beta(t-\tau_{i})}\right),
\end{equation*}
where $C_{i}$ are i.i.d. exponentially distributed independent of $\mathcal{F}_{\tau_{i}-}$. 
If we let 
\begin{equation*}
U_{t}:=ue^{-\beta t}+\sum\nolimits_{i:0<\tau_{i}<t}C_{i}e^{-\beta(t-\tau_{i})}, 
\end{equation*}
then $U_{t}$
satisfies the dynamics \eqref{UDynamics} with $\eta(u):=\beta u$. 
When $\lambda(\cdot)$ is linear, it is called linear Hawkes process, named after Hawkes \cite{Hawkes}.
The linear Hawkes process can be studied via immigration-birth representation, see e.g. Hawkes and Oakes \cite{HawkesII}.
When $\lambda(\cdot)$ is nonlinear, the Hawkes process is said to be nonlinear and the nonlinear Hawkes process
was first introduced by Br\'{e}maud and Massouli\'{e} \cite{Bremaud}.
The limit theorems for linear and nonlinear Hawkes processes have been studied in e.g. 
\cite{Bacry,Bordenave,ZhuI,ZhuII,ZhuMDP,ZhuCLT,ZhuCIR,Karabash,GZI,GZII,GZIII}.
The applications of Hawkes processes to insurance have been studied in e.g.
\cite{DassiosII,Stabile,ZhuRuin,GZII,ChengSeol}. As a by-product and corollary of the ruin probabilities results
obtained in this paper, the first-passage time for nonlinear Hawkes process with exponential kernel
and exponentially distributed jump sizes is therefore also analytically tractable, which is of independent interest
and is a new contribution to the theory of Hawkes processes.

The paper is organized as follows. In Section~\ref{MainSection}, we will derive
the ruin probability for the wealth process $U_{t}$ in closed-forms in Section~\ref{sec:ruin}, 
and we will also study the moments of the dividends in Section~\ref{sec:expected:div}, 
first and second moments of the wealth process in Section~\ref{sec:first:two:moments}, Laplace transform
of the ruin time in Section~\ref{sec:Laplace} and expected ruin time in Section~\ref{sec:ruin:time}. 
Numerical examples will be provided in Section~\ref{sec:numerical}.
The proofs will be provided in Appendix~\ref{ProofSection}.

\section{Main Results}\label{MainSection}

\subsection{Ruin Probability}\label{sec:ruin}

We obtain the close-form expression for the ruin probability
for the wealth process $U_{t}$ in \eqref{UDynamics} for a state-dependent dual risk model
as follows.

\begin{theorem}\label{RuinThm}
Assume that $\int_{0}^{\infty}\frac{\lambda(v)}{\eta(v)}e^{\gamma v-\int_{0}^{v}\frac{\lambda(w)}{\eta(w)}dw}dv$
exists and is finite.
Then, the ruin probability $\psi(u)=\mathbb{P}(\tau<\infty|U_{0}=u)$ is given by
\begin{equation}\label{ruin:main:formula}
\psi(u)=\mathbb{P}(\tau<\infty)=\frac{\int_{u}^{\infty}\frac{\lambda(v)}{\eta(v)}e^{\gamma v-\int_{0}^{v}\frac{\lambda(w)}{\eta(w)}dw}dv}
{\int_{0}^{\infty}\frac{\lambda(v)}{\eta(v)}e^{\gamma v-\int_{0}^{v}\frac{\lambda(w)}{\eta(w)}dw}dv}.
\end{equation}
\end{theorem}

One way to interpret the formula \eqref{ruin:main:formula} for the ruin probability in Theorem~\ref{RuinThm} is to write 
the formula \eqref{ruin:main:formula} as
\begin{equation}
\psi(u)=\frac{\mathbb{E}[e^{\gamma V}1_{V\geq u}]}{\mathbb{E}[e^{\gamma V}]}=\hat{\mathbb{P}}(V\geq u),
\end{equation}
where $V$ is a positive random variable with probability density function $\frac{\lambda(v)}{\eta(v)}e^{-\int_{0}^{v}\frac{\lambda(w)}{\eta(w)}dw}$,
and $\hat{\mathbb{P}}$ is a probability measure defined via the Radon-Nikodym derivative $\frac{d\hat{\mathbb{P}}}{d\mathbb{P}}=\frac{e^{\gamma V}}{\mathbb{E}[e^{\gamma V}]}$.

Figure~\ref{WealthI} provides illustrations of the wealth process against time
till the time when the company is ruined. In Figure~\ref{WealthI}(a), $\eta(z)$ is a constant and we can
see that the wealth process always decays with the constant rate. In Figure~\ref{WealthI}(b), $\eta(z)$ is linear in $z$, i.e. $\eta(z)=\alpha+\beta z$,
for some $\alpha,\beta>0$ and the wealth process decays exponentially and might get ruined. 
A nonparametric approach to the decay function $\eta(z)$ gives us more flexibility. 

\begin{figure}[htb]
\centering
\subfigure[Constant $\eta(\cdot)$]{
\includegraphics[scale=0.45]{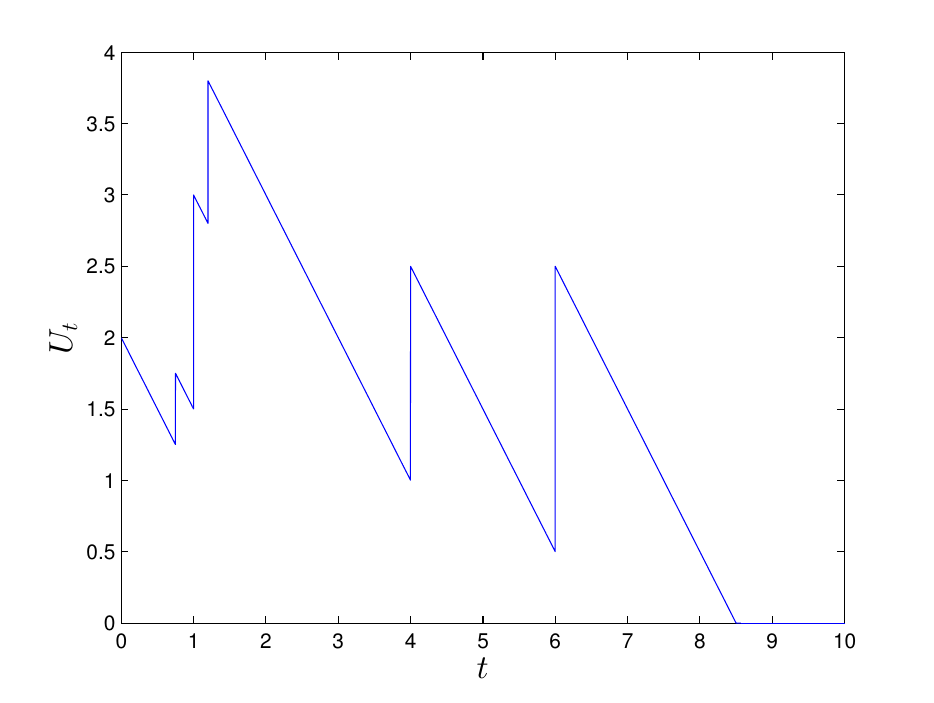}}
\subfigure[Linear $\eta(\cdot)$]{
\includegraphics[scale=0.45]{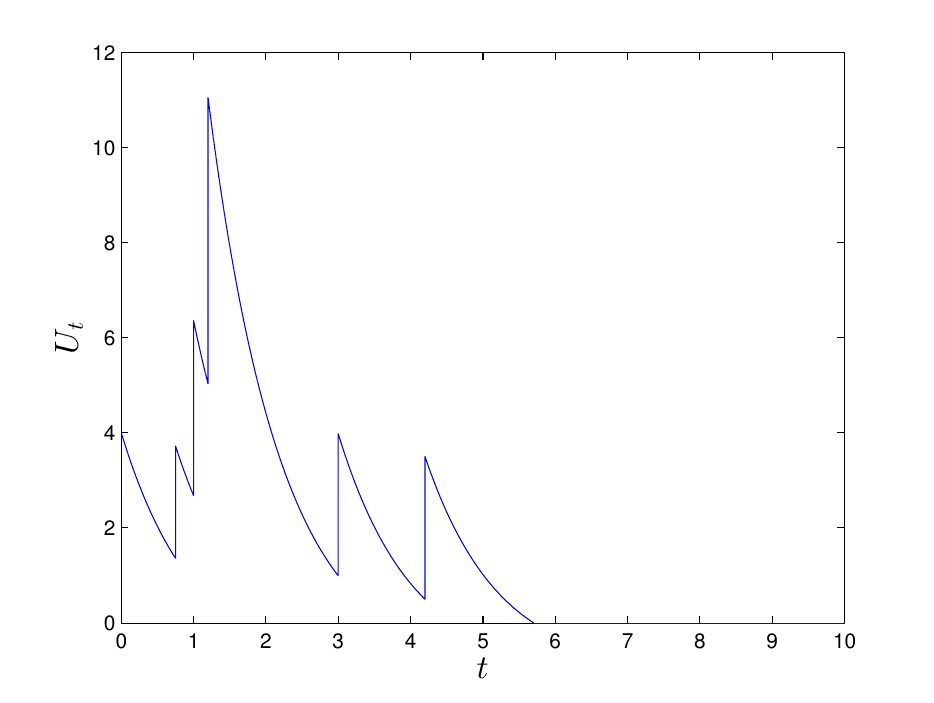}}
\caption{An illustration of the wealth process against time till the company is ruined.}
\label{WealthI}
\end{figure}

\subsection{Moments of the Dividends}\label{sec:expected:div}

One can also study the single dividend payment problem for the state-dependent dual risk model
with the wealth process $U_{t}$ in \eqref{UDynamics}. 
Let $b>U_{0}$ be the barrier
of the dividend. For the first time that the wealth process $U_{t}$ goes above the barrier $b$, say at
the first-passage time $\tau_{b}:=\inf\{t>0:U_{t}\geq b\}$, 
a dividend of the amount $D=U_{\tau_{b}}-b$ is paid out. No dividend is paid out
if the company is ruined before ever hitting the barrier $b$.
We are interested to compute that all the moments of the dividend to be paid out $\mathbb{E}[D^{k}1_{\tau_{b}<\tau}]$
for any $k\in\mathbb{N}$.

Note that under the assumption that the $C_{i}$'s are i.i.d. exponentially distributed with parameter $\gamma>0$,
from the memoryless property of exponential distribution, $U_{\tau_{b}}-b$ is also exponentially
distributed with parameter $\gamma>0$. Therefore, by applying the formula for the $k$-th moment of an exponential random variable, we have
for any $k\in\mathbb{N}$,
\begin{equation}
\mathbb{E}[D^{k}1_{\tau_{b}<\tau}]=\frac{k!}{\gamma^{k}}\mathbb{P}(\tau_{b}<\tau).
\end{equation}
Hence, the problem reduces to compute the probability that the dividend will be paid out before
the company is ruined and we have the following result.

\begin{theorem}\label{DividendThm}
Assume that $\int_{0}^{\infty}\frac{\lambda(v)}{\eta(v)}e^{\gamma v-\int_{0}^{v}\frac{\lambda(w)}{\eta(w)}dw}dv$
exists and is finite.
Then, the probability $\phi(u,b):=\mathbb{P}(\tau_{b}<\tau|U_{0}=u)$ is given by
\begin{equation}
\phi(u,b)=\frac{\int_{0}^{u}\frac{\lambda(v)}{\eta(v)}e^{\gamma v-\int_{0}^{v}\frac{\lambda(w)}{\eta(w)}dw}dv}
{\int_{0}^{\infty}\gamma e^{-\gamma c}\int_{0}^{b+c}\frac{\lambda(v)}{\eta(v)}
e^{\gamma v-\int_{0}^{v}\frac{\lambda(w)}{\eta(w)}dw}dvdc},
\end{equation}
and the $k$-th moment is given by
\begin{equation}
\mathbb{E}[D^{k}1_{\tau_{b}<\tau}]
=\frac{k!}{\gamma^{k}}\frac{\int_{0}^{u}\frac{\lambda(v)}{\eta(v)}e^{\gamma v-\int_{0}^{v}\frac{\lambda(w)}{\eta(w)}dw}dv}
{\int_{0}^{\infty}\gamma e^{-\gamma c}\int_{0}^{b+c}\frac{\lambda(v)}{\eta(v)}
e^{\gamma v-\int_{0}^{v}\frac{\lambda(w)}{\eta(w)}dw}dvdc}.
\end{equation}
\end{theorem}

In Theorem~\ref{DividendThm}, a single dividend payment is considered. 
One can also study multiple dividend payments as follows:
\begin{equation}
\tau_{b}^{(1)}:=\inf\{t>0:U_{t}>b\},
\qquad
\tau_{b}^{(i)}:=\inf\{t>\tau_{b}^{(i)}:U_{t}>b\},\qquad i\geq 2.
\end{equation}
Then $\tau_{b}^{(i)}$ is the $i$th payment of the dividend if $\tau_{b}^{(i)}<\tau$.
Let $N$ be the total number of dividends to be paid out before the ruin occurs and 
$\sum_{i=1}^{N}D_{i}$ be the total value of dividends to be paid out before the ruin occurs.
Recall $\phi(u,b)=\mathbb{P}(\tau_{b}<\tau|U_{0}=u)$ with closed-form formulas computed in Theorem~\ref{DividendThm}.
It is easy to see that
\begin{align}
&\mathbb{P}(N=0)=1-\phi(u,b),
\\
&\mathbb{P}(N=n)=\phi(u,b)\phi(b,b)^{n-1}(1-\phi(b,b)),\qquad n\geq 1.
\end{align}
Therefore, 
\begin{equation}
\mathbb{E}\left[\sum_{i=1}^{N}D_{i}\right]
=\frac{1}{\gamma}\mathbb{E}[N]
=\frac{1}{\gamma}\frac{\phi(u,b)}{1-\phi(b,b)}.
\end{equation}
One can also compute the Laplace transform of the total amount of dividends to be paid out.
For any $\theta>0$ such that $\frac{\gamma}{\gamma+\theta}\phi(b,b)<1$, we have
\begin{align}
\mathbb{E}\left[e^{-\theta\sum_{i=1}^{N}D_{i}}\right]
&=\mathbb{E}\left[\mathbb{E}\left[e^{-\theta\sum_{i=1}^{N}D_{i}}\big|N\right]\right]
\\
&=\mathbb{E}\left[\left(\frac{\gamma}{\gamma+\theta}\right)^{N}\right]
\nonumber
\\
&=1-\phi(u,b)+\phi(u,b)\sum_{n=1}^{\infty}\left(\frac{\gamma}{\gamma+\theta}\right)^{n}\phi(b,b)^{n-1}(1-\phi(b,b))
\nonumber
\\
&=1-\phi(u,b)+\phi(u,b)(1-\phi(b,b))\frac{\frac{\gamma}{\gamma+\theta}}{1-\frac{\gamma}{\gamma+\theta}\phi(b,b)}.
\nonumber
\end{align}

\subsection{First and Second Moments of the Wealth Process}\label{sec:first:two:moments}

We are also interested to study the first and second moments of the wealth process $U_{t}$ given in \eqref{UDynamics}.
Note that since the wealth process is defined only up to the ruin time $\tau$, we should evaluate
$\mathbb{E}[U_{t\wedge\tau}]$ and $\mathbb{E}[U_{t\wedge\tau}^{2}]$, which in general is a challenge
to compute since it will require us to know explicitly the distribution of the ruin time. 
We derive the first and second moments of the wealth process $U_{t}$ for a special case instead.
Let $\eta(u)\equiv\rho+\mu u$ and $\lambda(u)=\alpha+\beta u$, for some $\alpha,\beta\geq 0$, i.e. 
\begin{equation}
dU_{t}=-(\rho+\mu U_{t})dt+dJ_{t}
\end{equation}
where $J_{t}=\sum_{i=1}^{N_{t}}C_{i}$, 
where $N_{t}$ is a simple point process
with intensity $\lambda(U_{t-})=\alpha+\beta U_{t-}$  and $C_{i}$ are i.i.d.
with distribution $Q(dc)$.

In this case, $\tau\geq T_{0}$, where $T_{0}$ is the time that the ODE
\begin{equation}
du_{t}=-(\rho+\mu u_{t})dt,\qquad u_{t}=u,
\end{equation}
hits zero. It is easy to solve the above ODE and get
\begin{equation}
u_{t}=\left(\frac{\rho}{\mu}+u\right)e^{-\mu t}-\frac{\rho}{\mu},
\qquad
T_{0}=\frac{1}{\mu}\log\left(1+\frac{\mu u}{\rho}\right).
\end{equation}
Then, for any $t<T_{0}$, $t\wedge\tau=t$.
We have the following proposition that computes
the first two moments of the wealth process $U_{t}$.

\begin{proposition}\label{FirstSecondProp}
For any $t<\frac{1}{\mu}\log\left(1+\frac{\mu u}{\rho}\right)$,
\begin{equation}
\mathbb{E}[U_{t}]=\left(\frac{\rho-\alpha\mathbb{E}^{Q}[c]}{\mu-\beta\mathbb{E}^{Q}[c]}+u\right)e^{-(\mu-\beta\mathbb{E}^{Q}[c])t}
-\frac{\rho-\alpha\mathbb{E}^{Q}[c]}{\mu-\beta\mathbb{E}^{Q}[c]},
\end{equation}
and
\begin{align}
\mathbb{E}[U_{t}^{2}]
&=ue^{-2\beta(\mathbb{E}^{Q}[c]-\mu)t}
+\alpha\mathbb{E}^{Q}[c^{2}]
\frac{1-e^{-2\beta(\mathbb{E}^{Q}[c]-\mu)t}}{2\beta(\mathbb{E}^{Q}[c]-\mu)}
\nonumber
\\
&\qquad
-(2(\alpha\mathbb{E}^{Q}[c]-\rho)+\beta\mathbb{E}[c^{2}])\frac{\rho-\alpha\mathbb{E}^{Q}[c]}{\mu-\beta\mathbb{E}^{Q}[c]}
\frac{1-e^{-2\beta(\mathbb{E}^{Q}[c]-\mu)t}}{2\beta(\mathbb{E}^{Q}[c]-\mu)}
\nonumber
\\
&\qquad
+(2(\alpha\mathbb{E}^{Q}[c]-\rho)+\beta\mathbb{E}[c^{2}])\left(\frac{\rho-\alpha\mathbb{E}^{Q}[c]}{\mu-\beta\mathbb{E}^{Q}[c]}+u\right)
\frac{e^{-\beta(\mathbb{E}^{Q}[c]-\mu)t}-e^{-2\beta(\mathbb{E}^{Q}[c]-\mu)t}}{\beta(\mathbb{E}^{Q}[c]-\mu)}.
\nonumber
\end{align}
\end{proposition}

\subsection{Laplace Transform of Ruin Time}\label{sec:Laplace}

In the ruin theory of the dual risk models, it is of great interest to study the Laplace transform
of the ruin time,
\begin{equation}
\psi(u,\delta)=\mathbb{E}\left[e^{-\delta\tau}1_{\tau<\infty}\right],
\end{equation}
where $\delta>0$, and in this section, we aim to derive the Laplace transform
for the ruin time for the wealth process $U_{t}$ in \eqref{UDynamics}.
Note that $\psi(u,\delta)$ can also be interpreted as a perpetual digit option, 
with payoff $1$ dollar at the time of ruin, with discount coefficient $\delta>0$, which can be taken as the risk-free rate.

\begin{theorem}\label{LaplaceThm}
Assume that the equation 
\begin{align}
&\lambda(u)\eta(u)f''(u)
+[\lambda(u)\eta'(u)+\lambda^{2}(u)-\gamma\eta(u)\lambda(u)-\lambda'(u)\eta(u)+\delta\lambda(u)]f'(u)\label{f:eqn}
\\
&\qquad\qquad\qquad\qquad
-(\gamma\lambda(u)+\lambda'(u))\delta f(u)=0
\nonumber
\end{align}
has a uniformly bounded positive solution $f(u)$ that satisfies $f(\infty)=0$. 
Then, we have $\psi(u,\delta)=f(u)/f(0)$.
\end{theorem}

\subsection{Expected Ruin Time}\label{sec:ruin:time}

Previously, we computed the ruin probability $\mathbb{P}(\tau<\infty)$ in Theorem~\ref{RuinThm} under certain technical assumptions.
Note that when $\mathbb{P}(\tau<\infty)<1$, we have $\mathbb{E}[\tau]=\infty$. 
In the case that the ruin occurs with probability one, i.e., $\mathbb{P}(\tau<\infty)=1$, we can also 
compute that expected time that the ruin occurs as follows. 

\begin{theorem}\label{ExpectedThm}
Assume that $\mathbb{P}(\tau<\infty)=1$ and let us define
\begin{align}
f(u)&:=\int_{0}^{u}\frac{\lambda(v)}{\eta(v)}
\int_{0}^{v}[-\lambda'(w)-\gamma\lambda(w)]\frac{1}{\lambda(w)^{2}}
e^{\gamma(v-w)-\int_{w}^{v}\frac{\lambda(r)}{\eta(r)}dr}dwdv
\\
&\qquad
+g(0)\int_{0}^{u}\frac{\eta(0)}{\eta(v)}\frac{\lambda(v)}{\lambda(0)}e^{\gamma v}
e^{-\int_{0}^{v}\frac{\lambda(w)}{\eta(w)}dw}dv,
\nonumber
\end{align}
where
\begin{equation}\label{g0}
g(0):=\frac{1+\lambda(0)\int_{0}^{\infty}\int_{0}^{c}\frac{\lambda(v)}{\eta(v)}
\int_{0}^{v}[-\lambda'(w)-\gamma\lambda(w)]\frac{1}{\lambda(w)^{2}}
e^{\gamma(v-w)-\int_{w}^{v}\frac{\lambda(r)}{\eta(r)}dr}\gamma e^{-\gamma(c-u)}dwdvdc}
{\eta(0)-\lambda(0)\int_{0}^{\infty}\int_{0}^{c}\frac{\eta(0)}{\eta(v)}\frac{\lambda(v)}{\lambda(0)}e^{\gamma v}
e^{-\int_{0}^{v}\frac{\lambda(w)}{\eta(w)}dw}\gamma e^{-\gamma(c-u)}dvdc}.
\end{equation}
Assume that $\sup_{0<u<\infty}f(u)<\infty$. Then, $\mathbb{E}[\tau]=f(u)$.
\end{theorem}


\begin{figure}[htb]
\centering
\subfigure[$\lambda(u)=(\alpha u^{\beta}+\gamma)\eta(u)$]{
\includegraphics[scale=0.45]{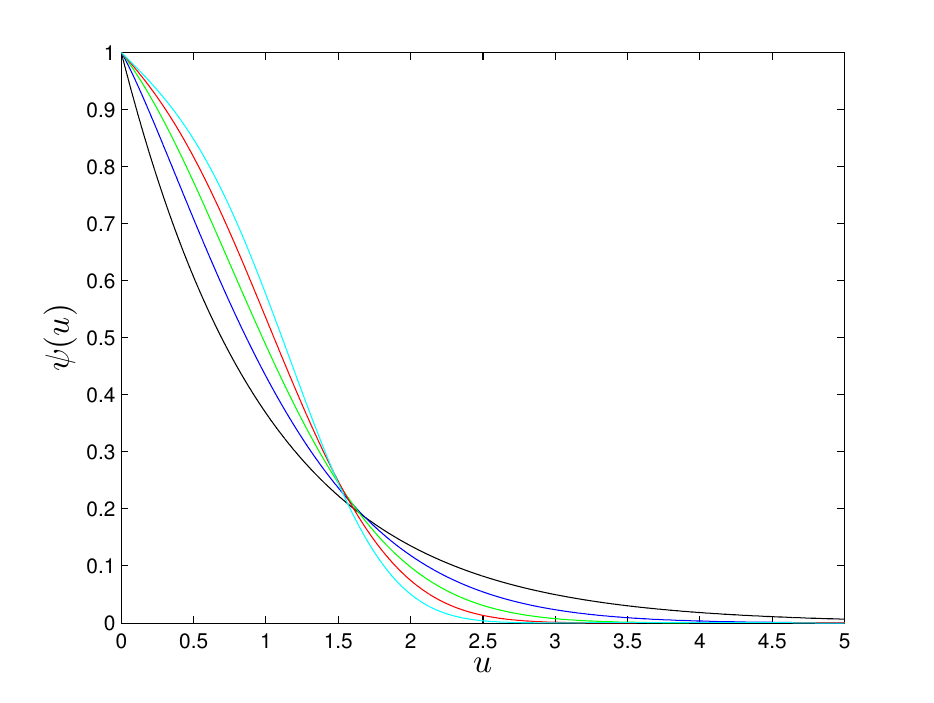}}
\subfigure[$\lambda(u)=(\gamma+\frac{\beta}{1+u})\eta(u)$]{
\includegraphics[scale=0.45]{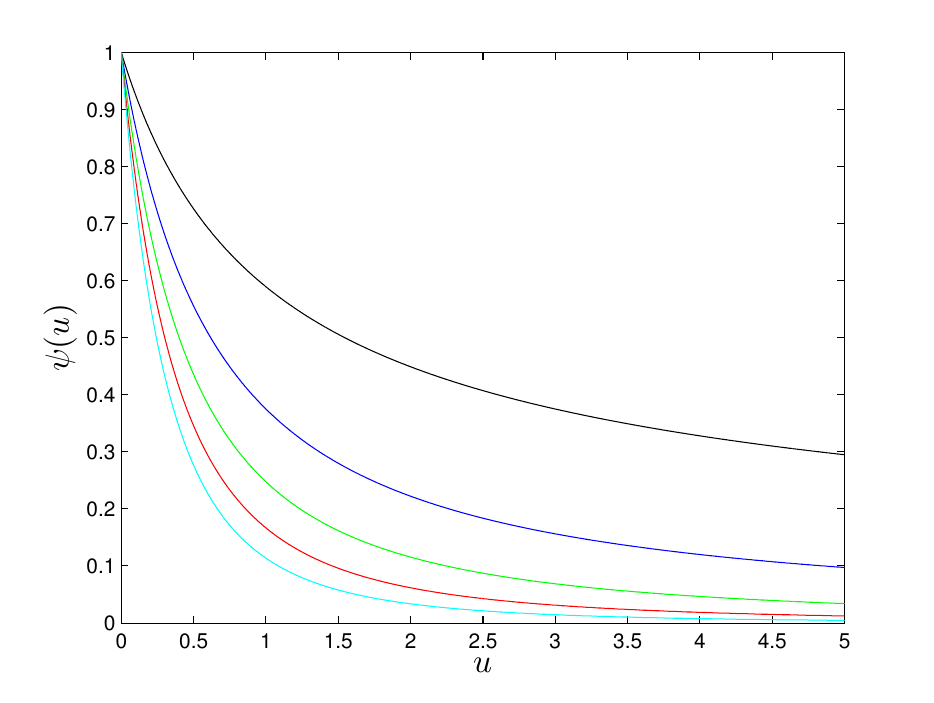}}
\caption{(a) Illustration of the ruin probability $\psi(u)$ against the initial wealth $u$ when
$\lambda(u)=(\alpha u^{\beta}+\gamma)\eta(u)$. The black, blue, green, red and cion lines denote the cases
when $\beta=0.0$, $0.5$, $1.0$, $1.5$ and $2.0$. The $\alpha$ and $\gamma$ are fixed to be $1.0$.
We can see from the plot that when $\beta=0.0$, the ruin probability exponentially decays in the initial wealth. 
Otherwise, the shape of decay is not exponential.
(b) Illustration of the ruin probability $\psi(u)$ against the initial wealth $u$ when
$\lambda(u)=(\gamma+\frac{\beta}{1+u})\eta(u)$. The black, blue, green, red and cion lines denote the cases
when $\beta=1.5$, $2.0$, $2.5$, $3.0$ and $3.5$. The $\gamma$ is fixed to be $1.0$.
The ruin probability decays polynomially against the initial wealth.}
\label{RuinProb}
\end{figure}

\begin{table}[htb]
\centering 
\begin{tabular}{|c|c|c|c|c|c|} 
\hline 
$\psi(u)$ & $u=1$ & $u=2$ & $u=3$ & $u=4$ & $u=5$ 
\\
\hline
$\beta=0.0$ & 0.3679 & 0.1353 & 0.0498 & 0.0183 & 0.0067
\\
$\beta=0.5$ & 0.4325 & 0.1184 & 0.0233 & 0.0035 & 0.0004
\\
$\beta=1.0$ & 0.4867 & 0.0981 & 0.0076 & 0.0002 & 0.0000
\\
$\beta=1.5$ & 0.5343 & 0.0747 & 0.0013 & 0.0000 & 0.0000
\\
$\beta=2.0$ & 0.5756 & 0.0506 & 0.0001 & 0.0000 & 0.0000
\\
\hline 
\end{tabular}
\caption{Illustration of the ruin probability $\psi(u)$ when $\lambda(u)=(\alpha u^{\beta}+\gamma)\eta(u)$ for fixed $\alpha=\gamma=1.0$ and varying $\beta$ and $u$.}
\label{RuinTable} 
\end{table}

\begin{table}[htb]
\centering 
\begin{tabular}{|c|c|c|c|c|c|} 
\hline 
$\psi(u)$ & $u=1$ & $u=2$ & $u=3$ & $u=4$ & $u=5$ 
\\
\hline
$\beta=1.5$ & 0.5893 & 0.4491 & 0.3750 & 0.3280 & 0.2948
\\
$\beta=2.0$ & 0.3750 & 0.2222 & 0.1562 & 0.1200 & 0.0972
\\
$\beta=2.5$ & 0.2475 & 0.1155 & 0.0688 & 0.0465 & 0.0340
\\
$\beta=3.0$ & 0.1667 & 0.0617 & 0.0312 & 0.0187 & 0.0123
\\
$\beta=3.5$ & 0.1136 & 0.0336 & 0.0145 & 0.0077 & 0.0046
\\
\hline 
\end{tabular}
\caption{Illustration of the ruin probability $\psi(u)$ when $\lambda(u)=(\gamma+\frac{\beta}{1+u})\eta(u)$ for fixed $\gamma=1.0$ and varying $\beta$ and $u$.}
\label{RuinPolyTable} 
\end{table}

\subsection{Numerical Examples}\label{sec:numerical}

In this section, we first illustrate the ruin probability $\psi(u)$ obtained in Theorem~\ref{RuinThm}
by some numerical examples. 

First, we consider the example $\lambda(u)=(\alpha u^{\beta}+\gamma)\eta(u)$, for some constants $\alpha,\beta>0$. 
In this example, the ruin probability in \eqref{ruin:main:formula} can be explicitly computed as
\begin{align}
\psi(u)=\frac{\int_{u}^{\infty}(\alpha v^{\beta}+\gamma)e^{-\frac{\alpha}{\beta+1}v^{\beta+1}}dv}
{\int_{0}^{\infty}(\alpha v^{\beta}+\gamma)e^{-\frac{\alpha}{\beta+1}v^{\beta+1}}dv}
=\frac{e^{-\frac{\alpha}{\beta+1}u^{\beta+1}}
+\frac{\gamma}{\beta+1}(\frac{\alpha}{\beta+1})^{-\frac{1}{\beta+1}}
\Gamma(\frac{1}{\beta+1},\frac{\alpha u^{\beta+1}}{\beta+1})}
{1
+\frac{\gamma}{\beta+1}(\frac{\alpha}{\beta+1})^{-\frac{1}{\beta+1}}
\Gamma(\frac{1}{\beta+1},0)},
\end{align}
where $\Gamma(s,x):=\int_{x}^{\infty}t^{s-1}e^{-t}dt$ is the incomplete gamma function.
The summary statistics of the ruin probability $\psi(u)$
for the case $\lambda(u)=(\alpha u^{\beta}+\gamma)\eta(u)$ with fixed $\alpha=\gamma=1.0$
and $\beta=0.0$, $0.5$, $1.0$, $1.5$ and $2.0$ are given in Figure~\ref{RuinProb}(a)
and Table~\ref{RuinTable}.

Next, we consider the example $\lambda(u)=(\gamma+\frac{\beta}{1+u})\eta(u)$, for some constant $\beta>1$. 
In this example, the ruin probability in \eqref{ruin:main:formula} can be explicitly computed as
\begin{align}
\psi(u)=\frac{\int_{u}^{\infty}(\gamma+\frac{\beta}{1+v})e^{-\beta\log(v+1)}dv}
{\int_{0}^{\infty}(\gamma+\frac{\beta}{1+v})e^{-\beta\log(v+1)}dv}
=\frac{\gamma}{\gamma+\beta-1}\frac{1}{(1+u)^{\beta-1}}+\frac{\beta-1}{\gamma+\beta-1}\frac{1}{(1+u)^{\beta}}.
\end{align}
The summary statistics of the ruin probability $\psi(u)$
for the case $\lambda(u)=(\gamma+\frac{\beta}{1+u})\eta(u)$ with fixed $\gamma=1.0$ and
$\beta=1.5$, $2.0$, $2.5$, $3.0$ and $3.5$ are given in Figure~\ref{RuinProb}(b)
and Table~\ref{RuinPolyTable}. As we can see from Figure~\ref{RuinProb}(a)
and Figure~\ref{RuinProb}(b), the shape of the ruin probability $\psi(u)$ in terms
of the initial wealth $u$ is not necessarily exponential. It exhibits a rich class of behaviors
as we vary the parameter $\beta$. Therefore, the state-dependent dual risk model we have built
is much more flexible and robust than many of the classical dual risk models in the literature.


Next, we illustrate the $k$-th moment of the dividend payment $\mathbb{E}[D^{k}1_{\tau_{b}<\tau}]$
in Theorem~\ref{DividendThm} by some numerical examples. 
It follows from Theorem~\ref{DividendThm} that for any $k\in\mathbb{N}$, 
\begin{equation}\label{obtain:1}
\mathbb{E}[D^{k}1_{\tau_{b}<\tau}]=\frac{k!f(u)}{\gamma^{k}\int_{0}^{\infty}f(b+c)\gamma e^{-\gamma c}dc},
\quad
\text{where}
\quad
f(u):=\int_{0}^{u}\frac{\lambda(v)}{\eta(v)}e^{\gamma v-\int_{0}^{v}\frac{\lambda(w)}{\eta(w)}dw}dv.
\end{equation}
Consider the example $\lambda(u)=(\gamma+\frac{\beta}{1+u})\eta(u)$, for some constant $\beta>1$. 
Then, we can compute that
\begin{equation}\label{obtain:2}
f(u)
=\frac{\gamma+\beta-1}{\beta-1}-\frac{\gamma}{\beta-1}\frac{1}{(1+u)^{\beta-1}}
-\frac{1}{(1+u)^{\beta}},
\end{equation}
and
\begin{equation}\label{plugging:1}
\int_{0}^{\infty}f(b+c)\gamma e^{-\gamma c}dc
=\frac{\gamma+\beta-1}{\beta-1}-\frac{\gamma}{\beta-1}\int_{0}^{\infty}\frac{\gamma e^{-\gamma c}}{(1+b+c)^{\beta-1}}dc
-\int_{0}^{\infty}\frac{\gamma e^{-\gamma c}}{(1+b+c)^{\beta}}dc.
\end{equation}
Applying integration by parts formula, we have
\begin{align}
\int_{0}^{\infty}\frac{\gamma e^{-\gamma c}}{(1+b+c)^{\beta}}dc
&=\frac{1}{1-\beta}\int_{0}^{\infty}\gamma e^{-\gamma c}d\left(\frac{1}{(1+b+c)^{\beta-1}}\right)
\nonumber
\\
&=\frac{1}{\beta-1}\frac{\gamma}{(1+b)^{\beta-1}}-\frac{1}{\beta-1}\int_{0}^{\infty}\frac{\gamma^{2}e^{-\gamma c}dc}{(1+b+c)^{\beta-1}}.\label{plugging:2}
\end{align}
By plugging \eqref{plugging:2} into \eqref{plugging:1}, we get
\begin{equation}\label{obtain:3}
\int_{0}^{\infty}f(b+c)\gamma e^{-\gamma c}dc
=\frac{\gamma+\beta-1}{\beta-1}-\frac{1}{\beta-1}\frac{\gamma}{(1+b)^{\beta-1}}.
\end{equation}

We illustrate the $k$-th moment of the dividend payment $\mathbb{E}[D^{k}1_{\tau_{b}<\tau}]$
for the case $\lambda(u)=(\gamma+\frac{\beta}{1+u})\eta(u)$
with fixed $\gamma=1.0$, $b=6.0$ and
$\beta=1.5$, $2.0$, $2.5$, $3.0$ and $3.5$ and $k=1,2$. 
It follows from \eqref{obtain:1}, \eqref{obtain:2} and \eqref{obtain:3} that
\begin{equation}
\mathbb{E}[D^{k}1_{\tau_{b}<\tau}]=\frac{k!}{\gamma^{k}}\frac{\frac{\gamma+\beta-1}{\beta-1}-\frac{\gamma}{\beta-1}\frac{1}{(1+u)^{\beta-1}}
-\frac{1}{(1+u)^{\beta}}}
{\frac{\gamma+\beta-1}{\beta-1}-\frac{1}{\beta-1}\frac{\gamma}{(1+b)^{\beta-1}}}.
\end{equation}
The summary statistics for $\mathbb{E}[D^{k}1_{\tau_{b}<\tau}]$, $k=1,2$,
for the case $\lambda(u)=(\gamma+\frac{\beta}{1+u})\eta(u)$ with fixed $\gamma=1.0$ and
$\beta=1.5$, $2.0$, $2.5$, $3.0$ and $3.5$ are given in Figure~\ref{ExpDividend}
and Table~\ref{ExpDividendTable}. We can see that the values for the second moments
are higher than the first moments, and $\mathbb{E}[D^{k}1_{\tau_{b}<\tau}]$
is increasing in both $\beta$, which is related to the innovation rate, 
and $u$, which is the initial wealth. That is consistent with the intuition.

\begin{figure}[htb]
\centering
\subfigure[1st moment]{
\includegraphics[scale=0.5]{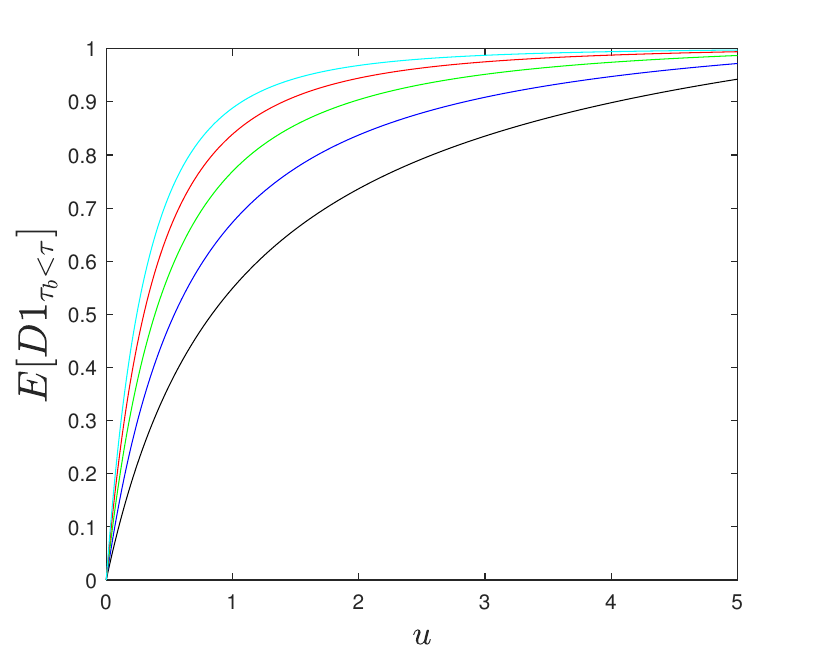}}
\subfigure[2nd moment]{
\includegraphics[scale=0.5]{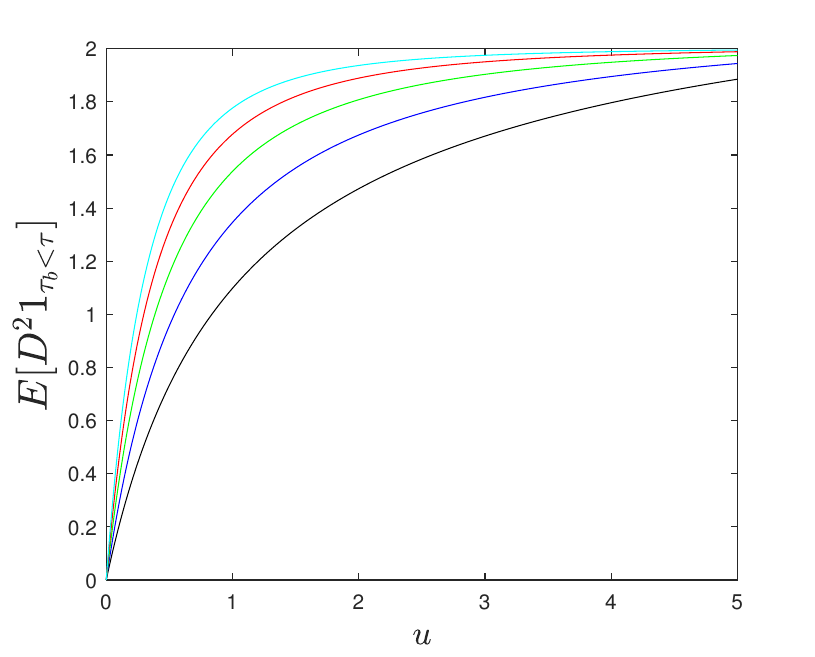}}
\caption{Illustration of $\mathbb{E}[D1_{\tau_{b}<\tau}]$ and $\mathbb{E}[D^{2}1_{\tau_{b}<\tau}]$ against the initial wealth $u$ when
$\lambda(u)=(\gamma+\frac{\beta}{1+u})\eta(u)$. The black, blue, green, red and cion lines denote the cases
when $\beta=1.5$, $2.0$, $2.5$, $3.0$ and $3.5$. The $\gamma$ is fixed to be $1.0$.}
\label{ExpDividend}
\end{figure}

\begin{table}[htb]
\centering 
\begin{tabular}{|c|c|c|c|c|c|} 
\hline 
$\mathbb{E}[D1_{\tau_{b}<\tau}]$ & $u=1$ & $u=2$ & $u=3$ & $u=4$ & $u=5$ 
\\
\hline
$\beta=1.5$ & 0.5491 & 0.7365 & 0.8355 & 0.8984 & 0.9427
\\
$\beta=2.0$ & 0.6731& 0.8376 & 0.9087 & 0.9477 & 0.9722
\\
$\beta=2.5$ & 0.7691 & 0.9041 & 0.9518 & 0.9745 & 0.9873
\\
$\beta=3.0$ & 0.8390 & 0.9447 & 0.9754 & 0.9881 & 0.9944
\\
$\beta=3.5$ & 0.8883 & 0.9685 & 0.9877 & 0.9945 & 0.9976
\\
\hline 
\hline 
$\mathbb{E}[D^{2}1_{\tau_{b}<\tau}]$ & $u=1$ & $u=2$ & $u=3$ & $u=4$ & $u=5$ 
\\
\hline
$\beta=1.5$ & 1.0982 & 1.4731 & 1.6711 & 1.7969 & 1.8854
\\
$\beta=2.0$ & 1.3462 & 1.6752 & 1.8173 & 1.8954 & 1.9444
\\
$\beta=2.5$ & 1.5382 & 1.8081 & 1.9036 & 1.9491 & 1.9746
\\
$\beta=3.0$ & 1.6781 & 1.8894 & 1.9508 & 1.9761 & 1.9888
\\
$\beta=3.5$ & 1.7766 & 1.9371 & 1.9753 & 1.9891 & 1.9952
\\
\hline 
\end{tabular}
\caption{Illustration of the first two moments of dividend payment $\mathbb{E}[D^{k}1_{\tau_{b}<\tau}]$, $k=1,2$, 
when $\lambda(u)=(\gamma+\frac{\beta}{1+u})\eta(u)$, for some constant $\beta>1$ for fixed $\gamma=1.0$ and $b=6.0$ and varying $\beta$ and $u$.}
\label{ExpDividendTable} 
\end{table}

\section*{Acknowledgements}
The author is grateful to an anonymous referee for helpful comments and suggestions
that have greatly improved the quality of the paper.
Lingjiong Zhu is grateful to the support from the National Science Foundation via the award
NSF-DMS-1613164.

\appendix

\section{Proofs}\label{ProofSection}

\begin{proof}[Proof of Theorem~\ref{RuinThm}]
From \eqref{UDynamics}, the infinitesimal generator of the wealth process $U_{t}$ \eqref{UDynamics} can be written as
\begin{equation}
\mathcal{A}f(u)=-\eta(u)\frac{\partial f}{\partial u}
+\lambda(u)\int_{0}^{\infty}[f(u+c)-f(u)]\gamma e^{-\gamma c}dc.
\end{equation}
Let us find a function $f$ such that $\mathcal{A}f=0$, which is equivalent to
\begin{equation}\label{diffIII}
-\eta(u)f'(u)-\lambda(u)f(u)+\lambda(u)\int_{u}^{\infty}f(c)\gamma e^{-\gamma(c-u)}dc=0.
\end{equation}
Differentiating \eqref{diffIII} with respect to $u$, we get
\begin{align}\label{diffIV}
&-\eta'(u)f'(u)-\eta(u)f''(u)-\lambda'(u)f(u)-\lambda(u)f'(u)
+\lambda'(u)\int_{u}^{\infty}f(c)\gamma e^{-\gamma(c-u)}dc
\\
&\qquad\qquad
-\lambda(u)\gamma f(u)
+\lambda(u)\gamma\int_{u}^{\infty}f(c)\gamma e^{-\gamma(c-u)}dc=0.
\nonumber
\end{align}
Substituting \eqref{diffIII} into \eqref{diffIV}, we get
\begin{equation}
\eta(u)f''(u)=\left[\frac{\lambda'(u)}{\lambda(u)}\eta(u)+\gamma\eta(u)-\eta'(u)-\lambda(u)\right]f'(u).
\end{equation}
By letting $f'(u)=g(u)$, we have
\begin{equation}\label{gEqnII}
\frac{dg}{g}=\left(\frac{\lambda'(u)}{\lambda(u)}+\gamma-\frac{\eta'(u)}{\eta(u)}-\frac{\lambda(u)}{\eta(u)}\right)du,
\end{equation}
which implies that
\begin{equation}
g(u)=\frac{\lambda(u)}{\eta(u)}e^{\gamma u-\int_{0}^{u}\frac{\lambda(v)}{\eta(v)}dv}
\end{equation}
is a particular solution to \eqref{gEqnII}. Hence, $\mathcal{A}f(u)=0$ for
\begin{equation}
f(u):=\int_{0}^{u}\frac{\lambda(v)}{\eta(v)}e^{\gamma v-\int_{0}^{v}\frac{\lambda(w)}{\eta(w)}dw}dv.
\end{equation}
By our assumption, $f(\infty)$ exists and is finite and it is also clear
that for any $0\leq u\leq\infty$, $0\leq f(u)\leq f(\infty)<\infty$.
Hence, by optional stopping theorem (see e.g. \cite{Durrett}),
\begin{equation}
f(u)=\mathbb{E}_{u}[f(U_{\tau})]=f(0)\mathbb{P}(\tau<\infty)+f(\infty)\mathbb{P}(\tau=\infty),
\end{equation}
which implies that
\begin{equation}
\psi(u)=\mathbb{P}(\tau<\infty)=\frac{\int_{u}^{\infty}\frac{\lambda(v)}{\eta(v)}e^{\gamma v-\int_{0}^{v}\frac{\lambda(w)}{\eta(w)}dw}dv}
{\int_{0}^{\infty}\frac{\lambda(v)}{\eta(v)}e^{\gamma v-\int_{0}^{v}\frac{\lambda(w)}{\eta(w)}dw}dv}.
\end{equation}
This completes the proof.
\end{proof}


\begin{proof}[Proof of Theorem~\ref{DividendThm}]
Let us recall from the proof of Theorem~\ref{RuinThm} that for 
\begin{equation}
f(u)=\int_{0}^{u}\frac{\lambda(v)}{\eta(v)}e^{\gamma v-\int_{0}^{v}\frac{\lambda(w)}{\eta(w)}dw}dv,
\end{equation}
we have $\mathcal{A}f=0$ and by our assumption $f$ is uniformly bounded, 
where $\mathcal{A}$ is the infinitesimal generator of the wealth process $U_{t}$ \eqref{UDynamics}.
By optional stopping theorem (see e.g. \cite{Durrett}),
\begin{align}
f(u)&=\mathbb{E}_{u}[f(U_{\tau\wedge\tau_{b}})]
\\
&=f(0)\mathbb{P}(\tau<\tau_{b})+\int_{0}^{\infty}f(b+c)\gamma e^{-\gamma c}dc\mathbb{P}(\tau_{b}<\tau),
\nonumber
\end{align}
which implies that for any $k\in\mathbb{N}$
\begin{align}
\mathbb{E}[D^{k}1_{\tau_{b}<\tau}]&=\frac{k!}{\gamma^{k}}\mathbb{P}(\tau_{b}<\tau)
\\
&=\frac{k!}{\gamma^{k}}\frac{f(u)-f(0)}{\int_{0}^{\infty}f(b+c)\gamma e^{-\gamma c}dc-f(0)}
\nonumber
\\
&=\frac{k!}{\gamma^{k}}\frac{\int_{0}^{u}\frac{\lambda(v)}{\eta(v)}e^{\gamma v-\int_{0}^{v}\frac{\lambda(w)}{\eta(w)}dw}dv}
{\int_{0}^{\infty}\gamma e^{-\gamma c}\int_{0}^{b+c}\frac{\lambda(v)}{\eta(v)}
e^{\gamma v-\int_{0}^{v}\frac{\lambda(w)}{\eta(w)}dw}dvdc}.
\nonumber
\end{align}
The proof is complete.
\end{proof}


\begin{proof}[Proof of Proposition~\ref{FirstSecondProp}]
The infinitesimal generator of $U_{t}$ process is given by
\begin{equation}
\mathcal{A}f(u)=-(\rho+\mu u)f'(u)+(\alpha+\beta u)\int_{0}^{\infty}[f(u+c)-f(u)]Q(dc).
\end{equation}
Let $f(u)=u$, we get $\mathcal{A}u=-\rho-\mu u+\mathbb{E}^{Q}[c](\alpha+\beta u)$.
By Dynkin's formula (see e.g. \cite{Oksendal}),
\begin{align}
\mathbb{E}[U_{t}]
&=u+\int_{0}^{t}\mathbb{E}\left[-\rho-\mu U_{s}+\mathbb{E}^{Q}[c](\alpha+\beta U_{s})\right]ds
\\
&=u+(-\rho+\alpha\mathbb{E}^{Q}[c])t+(\beta\mathbb{E}^{Q}[c]-\mu)\int_{0}^{t}\mathbb{E}[U_{s}]ds,
\nonumber
\end{align}
which yields that
\begin{equation}\label{MeanSol}
\mathbb{E}[U_{t}]=\left(\frac{\rho-\alpha\mathbb{E}^{Q}[c]}{\mu-\beta\mathbb{E}^{Q}[c]}+u\right)e^{-(\mu-\beta\mathbb{E}^{Q}[c])t}
-\frac{\rho-\alpha\mathbb{E}^{Q}[c]}{\mu-\beta\mathbb{E}^{Q}[c]}.
\end{equation}
Let $f(u)=u^{2}$, we get $\mathcal{A}u^{2}=\alpha\mathbb{E}^{Q}[c^{2}]+(2(\alpha\mathbb{E}^{Q}[c]-\rho)+\beta\mathbb{E}[c^{2}])u
+2(\beta\mathbb{E}[c]-\mu)u^{2}$. By Dynkin's formula (see e.g. \cite{Oksendal}),
\begin{equation}
\mathbb{E}[U_{t}^{2}]
=u^{2}+\alpha\mathbb{E}^{Q}[c^{2}]t+(2(\alpha\mathbb{E}^{Q}[c]-\rho)+\beta\mathbb{E}[c^{2}])\int_{0}^{t}\mathbb{E}[U_{s}]ds
+2(\beta\mathbb{E}[c]-\mu)\int_{0}^{t}\mathbb{E}[U_{s}^{2}]ds,
\end{equation}
which implies the ODE:
\begin{equation}
\frac{d}{dt}\mathbb{E}[U_{t}^{2}]
=\alpha\mathbb{E}^{Q}[c^{2}]+(2(\alpha\mathbb{E}^{Q}[c]-\rho)+\beta\mathbb{E}[c^{2}])\mathbb{E}[U_{t}]
+2(\beta\mathbb{E}[c]-\mu)\mathbb{E}[U_{t}^{2}].
\end{equation}
This is a first-order linear ODE. Together with \eqref{MeanSol}, we get
\begin{align}
\mathbb{E}[U_{t}^{2}]&=ue^{-2\beta(\mathbb{E}^{Q}[c]-\mu)t}
+\int_{0}^{t}e^{-2\beta(\mathbb{E}^{Q}[c]-\mu)(t-s)}
\alpha\mathbb{E}^{Q}[c^{2}]ds
\\
&\qquad
+(2(\alpha\mathbb{E}^{Q}[c]-\rho)+\beta\mathbb{E}[c^{2}])
\int_{0}^{t}e^{-2\beta(\mathbb{E}^{Q}[c]-\mu)(t-s)}
\left(\frac{\rho-\alpha\mathbb{E}^{Q}[c]}{\mu-\beta\mathbb{E}^{Q}[c]}+u\right)e^{-(\mu-\beta\mathbb{E}^{Q}[c])s}ds
\nonumber
\\
&\qquad\qquad
-(2(\alpha\mathbb{E}^{Q}[c]-\rho)+\beta\mathbb{E}[c^{2}])
\int_{0}^{t}e^{-2\beta(\mathbb{E}^{Q}[c]-\mu)(t-s)}
\frac{\rho-\alpha\mathbb{E}^{Q}[c]}{\mu-\beta\mathbb{E}^{Q}[c]}ds
\nonumber
\\
&=
ue^{-2\beta(\mathbb{E}^{Q}[c]-\mu)t}
+\alpha\mathbb{E}^{Q}[c^{2}]
\frac{1-e^{-2\beta(\mathbb{E}^{Q}[c]-\mu)t}}{2\beta(\mathbb{E}^{Q}[c]-\mu)}
\nonumber
\\
&\qquad
-(2(\alpha\mathbb{E}^{Q}[c]-\rho)+\beta\mathbb{E}[c^{2}])\frac{\rho-\alpha\mathbb{E}^{Q}[c]}{\mu-\beta\mathbb{E}^{Q}[c]}
\frac{1-e^{-2\beta(\mathbb{E}^{Q}[c]-\mu)t}}{2\beta(\mathbb{E}^{Q}[c]-\mu)}
\nonumber
\\
&\qquad
+(2(\alpha\mathbb{E}^{Q}[c]-\rho)+\beta\mathbb{E}[c^{2}])\left(\frac{\rho-\alpha\mathbb{E}^{Q}[c]}{\mu-\beta\mathbb{E}^{Q}[c]}+u\right)
\frac{e^{-\beta(\mathbb{E}^{Q}[c]-\mu)t}-e^{-2\beta(\mathbb{E}^{Q}[c]-\mu)t}}{\beta(\mathbb{E}^{Q}[c]-\mu)}.
\nonumber
\end{align}
This completes the proof.
\end{proof}


\begin{proof}[Proof of Theorem~\ref{LaplaceThm}]
Assume that we have a uniformly bounded positive function $f$ such that $\mathcal{A}f=\delta f$.
Note that $\frac{f(U_{t})}{f(U_{0})}e^{-\int_{0}^{t}\frac{\mathcal{A}f}{f}(U_{s})ds}
=\frac{f(U_{t})}{f(U_{0})}e^{-\delta t}$ is a martingale.
By optional stopping theorem (see e.g. \cite{Durrett}), 
\begin{equation}
1=\mathbb{E}\left[\frac{f(U_{\tau})}{f(U_{0})}e^{-\int_{0}^{\tau}\frac{\mathcal{A}f}{f}(U_{s})ds}\right]
=\frac{f(0)}{f(u)}\mathbb{E}\left[e^{-\delta\tau}1_{\tau<\infty}\right]. 
\end{equation}
Therefore, 
\begin{equation}
\psi(u,\delta)=f(u)/f(0).
\end{equation}

Let us now try to find a function $f$ such that $\mathcal{A}f=\delta f$.
Note that $\mathcal{A}f=\delta f$ is equivalent to
\begin{equation}
-\eta(u)f'(u)-\lambda(u)f(u)+\lambda(u)\int_{0}^{\infty}f(u+c)\gamma e^{-\gamma c}dc=\delta f,
\end{equation}
which implies that
\begin{align}
&\lambda(u)\eta(u)f''(u)
+[\lambda(u)\eta'(u)+\lambda^{2}(u)-\gamma\eta(u)\lambda(u)-\lambda'(u)\eta(u)+\delta\lambda(u)]f'(u)
\\
&\qquad\qquad\qquad\qquad
-(\gamma\lambda(u)+\lambda'(u))\delta f(u)=0.
\nonumber
\end{align}
This completes the proof.
\end{proof}


\begin{proof}[Proof of Theorem~\ref{ExpectedThm}]
Recall that the infinitesimal generator of $U_{t}$ process is given by
\begin{equation}
\mathcal{A}f(u)=-\eta(u)\frac{\partial f}{\partial u}
+\lambda(u)\int_{0}^{\infty}[f(u+c)-f(u)]\gamma e^{-\gamma c}dc.
\end{equation}
Let us find a function $f$ such that $\mathcal{A}f=-1$. That is,
\begin{equation}\label{subI}
-\eta(u)f'(u)-\lambda(u)f(u)
+\lambda(u)\int_{u}^{\infty}f(c)\gamma e^{-\gamma(c-u)}dc=-1.
\end{equation}
Differentiating the equation \eqref{subI} with respect to $u$, we get
\begin{align}\label{subII}
&-\eta'(u)f'(u)-\eta(u)f''(u)-\lambda'(u)f(u)-\lambda(u)f'(u)
+\lambda'(u)\int_{u}^{\infty}f(c)\gamma e^{-\gamma(c-u)}dc
\\
&\qquad\qquad
-\lambda(u)\gamma f(u)+\lambda(u)\gamma\int_{u}^{\infty}f(c)\gamma e^{-\gamma(c-u)}dc=0.
\nonumber
\end{align}
Substituting \eqref{subI} into \eqref{subII}, we get
\begin{equation}
f''(u)+\left[\frac{\eta'(u)}{\eta(u)}+\frac{\lambda(u)}{\eta(u)}-\frac{\lambda'(u)}{\lambda(u)}-\gamma\right]f'(u)
=-\left[\frac{\lambda'(u)}{\lambda(u)}+\gamma\right]\frac{1}{\eta(u)}.
\end{equation}
Let $g(u)=f'(u)$, then $g(u)$ satisfies a first-order linear ODE:
\begin{equation}
g'(u)+\left[\frac{\eta'(u)}{\eta(u)}+\frac{\lambda(u)}{\eta(u)}-\frac{\lambda'(u)}{\lambda(u)}-\gamma\right]g(u)
=-\left[\frac{\lambda'(u)}{\lambda(u)}+\gamma\right]\frac{1}{\eta(u)},
\end{equation}
which yields the general solution
\begin{align}
g(u)&=\int_{0}^{u}\left[-\frac{\lambda'(v)}{\lambda(v)}-\gamma\right]\frac{1}{\eta(v)}
e^{-\int_{v}^{u}\left[\frac{\eta'(w)}{\eta(w)}+\frac{\lambda(w)}{\eta(w)}-\frac{\lambda'(w)}{\lambda(w)}-\gamma\right]dw}dv
\\
&\qquad\qquad\qquad
+g(0)e^{-\int_{0}^{u}\left[\frac{\eta'(v)}{\eta(v)}+\frac{\lambda(v)}{\eta(v)}-\frac{\lambda'(v)}{\lambda(v)}-\gamma\right]dv}
\nonumber
\\
&=\frac{\lambda(u)}{\eta(u)}
\int_{0}^{u}[-\lambda'(v)-\gamma\lambda(v)]\frac{1}{\lambda(v)^{2}}
e^{\gamma(u-v)-\int_{v}^{u}\frac{\lambda(w)}{\eta(w)}dw}dv
\nonumber
\\
&\qquad\qquad\qquad
+g(0)\frac{\eta(0)}{\eta(u)}\frac{\lambda(u)}{\lambda(0)}e^{\gamma u}
e^{-\int_{0}^{u}\frac{\lambda(v)}{\eta(v)}dv}.
\nonumber
\end{align}
Hence, we can choose $f(u)=\int_{0}^{u}g(v)dv$, which gives
\begin{align}
f(u)&=\int_{0}^{u}\frac{\lambda(v)}{\eta(v)}
\int_{0}^{v}[-\lambda'(w)-\gamma\lambda(w)]\frac{1}{\lambda(w)^{2}}
e^{\gamma(v-w)-\int_{w}^{v}\frac{\lambda(r)}{\eta(r)}dr}dwdv
\\
&\qquad
+g(0)\int_{0}^{u}\frac{\eta(0)}{\eta(v)}\frac{\lambda(v)}{\lambda(0)}e^{\gamma v}
e^{-\int_{0}^{v}\frac{\lambda(w)}{\eta(w)}dw}dv.
\nonumber
\end{align}
Next, let us determine $g(0)$. Recall that $g(0)=f'(0)$ and also notice that $f(0)=0$. 
Note that by letting $u=0$ in \eqref{subI}, we get
\begin{equation}
-\eta(0)g(0)
+\lambda(0)\int_{0}^{\infty}f(c)\gamma e^{-\gamma(c-u)}dc=-1,
\end{equation}
which implies that
\begin{align}
-1&=-\eta(0)g(0)+\lambda(0)\int_{0}^{\infty}\int_{0}^{c}\frac{\lambda(v)}{\eta(v)}
\int_{0}^{v}[-\lambda'(w)-\gamma\lambda(w)]\frac{1}{\lambda(w)^{2}}
\\
&\qquad\qquad\qquad\cdot
e^{\gamma(v-w)-\int_{w}^{v}\frac{\lambda(r)}{\eta(r)}dr}\gamma e^{-\gamma(c-u)}dwdvdc
\nonumber
\\
&\qquad\qquad
+g(0)\lambda(0)\int_{0}^{\infty}\int_{0}^{c}\frac{\eta(0)}{\eta(v)}\frac{\lambda(v)}{\lambda(0)}e^{\gamma v}
e^{-\int_{0}^{v}\frac{\lambda(w)}{\eta(w)}dw}\gamma e^{-\gamma(c-u)}dvdc.
\nonumber
\end{align}
Therefore, 
\begin{equation}
g(0)=\frac{1+\lambda(0)\int_{0}^{\infty}\int_{0}^{c}\frac{\lambda(v)}{\eta(v)}
\int_{0}^{v}[-\lambda'(w)-\gamma\lambda(w)]\frac{1}{\lambda(w)^{2}}
e^{\gamma(v-w)-\int_{w}^{v}\frac{\lambda(r)}{\eta(r)}dr}\gamma e^{-\gamma(c-u)}dwdvdc}
{\eta(0)-\lambda(0)\int_{0}^{\infty}\int_{0}^{c}\frac{\eta(0)}{\eta(v)}\frac{\lambda(v)}{\lambda(0)}e^{\gamma v}
e^{-\int_{0}^{v}\frac{\lambda(w)}{\eta(w)}dw}\gamma e^{-\gamma(c-u)}dvdc}.
\end{equation}

By Dynkin's formula (see e.g. \cite{Oksendal}), for any $K>0$,  
\begin{equation}
\mathbb{E}[f(U_{\tau\wedge K})]=f(u)+\mathbb{E}\left[\int_{0}^{\tau\wedge K}\mathcal{A}f(U_{s})ds\right]
=f(u)-\mathbb{E}[\tau\wedge K].
\end{equation}
By our assumption $\sup_{0<u<\infty}f(u)<\infty$ and $\tau<\infty$ a.s. Hence
as $K\rightarrow\infty$, from bounded convergence theorem, 
we have $\mathbb{E}[f(U_{\tau\wedge K})]\rightarrow\mathbb{E}[f(U_{\tau})]$
and by monotone convergence theorem, $\mathbb{E}[\tau\wedge K]\rightarrow\mathbb{E}[\tau]$.
Therefore, $\mathbb{E}[\tau]=f(u)-\mathbb{E}[f(U_{\tau})]$,
which implies that
\begin{align}
\mathbb{E}[\tau]&=\int_{0}^{u}\frac{\lambda(v)}{\eta(v)}
\int_{0}^{v}[-\lambda'(w)-\gamma\lambda(w)]\frac{1}{\lambda(w)^{2}}
e^{\gamma(v-w)-\int_{w}^{v}\frac{\lambda(r)}{\eta(r)}dr}dwdv
\\
&\qquad
+g(0)\int_{0}^{u}\frac{\eta(0)}{\eta(v)}\frac{\lambda(v)}{\lambda(0)}e^{\gamma v}
e^{-\int_{0}^{v}\frac{\lambda(w)}{\eta(w)}dw}dv.
\nonumber
\end{align}
where $g(0)$ is defined in \eqref{g0}. The proof is complete.
\end{proof}

\end{document}